\documentclass{article}
\usepackage{booktabs} % For formal tables
\usepackage[ruled]{algorithm2e} % For algorithms
\usepackage{enumitem}

\usepackage[letterpaper,top=1in,bottom=1in,left=1in,right=1in,marginparwidth=1.75cm]{geometry}
\usepackage{amsmath}
\usepackage{graphicx}
\usepackage{amssymb}
\usepackage{amsthm}
\usepackage{lineno}
\usepackage{tikz}
\usepackage{url}
\usepackage{standalone}
\usetikzlibrary{shapes,decorations}
\usetikzlibrary{arrows.meta}

\usepackage{todonotes}

\usepackage{comment}

\newtheorem{theorem}{Theorem}[section]
\newtheorem{definition}[theorem]{Definition}
\newtheorem{lemma}[theorem]{Lemma}

\newtheorem{proposition}[theorem]{Proposition}

\newtheorem{conjecture}[theorem]{Conjecture}

 \def\real{{\mathbb R}}

\newcommand{\NP}{\ensuremath{\mathsf{NP}}}

\newcommand{\Ptime}{\ensuremath{\mathsf{P}}}

\newcommand{\PPAD}{\ensuremath{\mathsf{PPAD}}}

\newcommand{\CLS}{\ensuremath{\mathsf{CLS}}}

\newcommand{\RP}{\ensuremath{\mathsf{RP}}}

\DeclareMathOperator{\poly}{poly}
\SetAlFnt{\small}
\SetAlCapFnt{\small}
\SetAlCapNameFnt{\small}
\SetAlCapHSkip{0pt}
\IncMargin{-\parindent}

\title{On the Complexity of Learning Nash Equilibria
}
\author{Oliver Biggar, Christos Papadimitriou, and Georgios Piliouras}
\date{}

\begin{document}

\begin{titlepage}

\maketitle

\begin{abstract}
\noindent
We know that the Nash equilibria of a game cannot be computed efficiently unless $\Ptime = \PPAD$.  But can they be {\em learned?}  Are there dynamics that (1) can be computed efficiently by the players at each strategy profile and (2) are guaranteed to converge to the Nash equilibria?  This is a question as ancient as the Nash equilibrium itself, and antedates by many decades current complexity considerations about it.  It was recently proved in \cite{milionis_impossibility_2023} that no such dynamics can exist in general; however, the game used in that proof is degenerate, and a strong assumption of uniform convergence to a continuum of Nash equilibria is employed. We point out that both assumptions are necessary for that proof, because Nash-convergent dynamics do exist, which converge to all Nash equilibria in non-degenerate games; in fact we describe %, we believe for the first time, 
two very different families of such dynamics.  However, both of these families are intractable to compute locally.  We formulate a complexity theoretic \emph{Impossibility Conjecture}: if a locally tractable Nash-convergent dynamic exists then $\Ptime=\PPAD$.  This is a novel kind of conjecture, combining topology and complexity, because the dynamic is only required to be tractable locally --- it could take exponentially many steps to converge, so long as the work done at each step is polynomial. We validate this conjecture for the two specific families of dynamics introduced, showing that their tractability would imply complexity collapses, namely $\NP=\RP$ and $\PPAD= \Ptime$. Next, we show that any locally tractable dynamic which possesses a locally tractable \emph{Lyapunov function} cannot exist unless $\PPAD = \CLS$. En route to these results, we settle the complexity of a problem of interest in its own right, namely finding a Nash equilibrium of a game that lies on a given affine subspace: it is polynomial for affine spaces of constant dimension and NP-hard for polynomially large dimensions. Finally, regarding the most general impossibility conjecture, not requiring a Lyapunov function, we provide a complexity-theoretic explanation why it may be difficult to settle: we introduce a \emph{Proving Game} to demonstrate that black-box reductions cannot distinguish between convergent and non-convergent dynamics in polynomial time. %Finally, we establish that the Impossibility Theorem for global Nash convergence of \cite{milionis_impossibility_2023} also holds for nondegenerate games, unless $PPAD = CLS$. 
Overall, our results suggest that the barrier to Nash learning is not the non-existence of a vector field for some games, but the intractability of locally computing this vector field, a novel and subtle form of intractability. And yet, the strongest form of this intractability remains open.
\end{abstract}

\end{titlepage}

\section{Introduction}
The Nash equilibrium, proposed by John F.~Nash in 1950 \cite{nash_equilibrium_1950}, is arguably the most influential idea in the history of economics.  It launched Game Theory, inspired Arrow and Debreu to prove the fundamental theorems of welfare economics, while its universality --- the fact that every game is guaranteed to have one --- enabled economists everywhere, while designing new markets, institutions, or mechanisms, to contemplate the equilibrium with the certainty, afforded by Nash's theorem, that such an equilibrium exists.

When game theory came under scrutiny by computer scientists in the late 1990s, certain serious computational shortcomings of the Nash equilibrium became clear. The Nash equilibrium was eventually proved intractable to compute in normal-form games even with two players \cite{daskalakis_complexity_2009,chen2006settling,etessami2010complexity}, while the multiplicity of Nash equilibria --- the other computationally unwieldy aspect of the Nash equilibrium concept --- appears to have no principled, computationally feasible remedy \cite{goldberg_complexity_2013}. A common defense of the Nash equilibrium by economists in the face of such computational criticism is that {\em ``despite all, the players will eventually get there.''}  Notice that this argument shifts the debate from complexity theory to the realm of {\em learning dynamics in games,} see Section 2 for a brief introduction; by {\em learning dynamics} we mean dynamics that are meant to model the reaction of agents to the experience of repeated play. Indeed, dozens of different learning dynamics have been proposed by game theorists, economists, and computer scientists over the decades (some even preceding Nash's theorem \cite{brown1949some}), in many cases explicitly motivated by this quest for Nash-convergent dynamics. In fact, Nash's own 1950 proof of existence \cite{nash_equilibrium_1950} starts by defining a particular dynamic, now often called {\em the Brown-von Neumann-Nash dynamic} \cite{sandholm2010population}, which, however, is known to cycle. Implicit in the requirements of a `sensible' learning dynamic is \emph{local tractability}, meaning that for any given game, the displacement $\phi(x)$ --- or the gradient $\dot \phi(x)$ for a continuous-time dynamic --- can be computed at any mixed strategy profile $x$ in time that is polynomial in the description of the game. However, to date there is no known locally tractable dynamic $\phi$ that is also \emph{Nash-convergent}, in that it is guaranteed to converge to a Nash equilibrium in all games and from all starting points. 

%that is both {\em locally tractable and Nash-convergent.}  A dynamic is locally tractable if, for any given game, the displacement $\phi(x)$ --- or the gradient $\dot \phi(x)$ for a continuous-time dynamic --- can be computed at any mixed strategy profile $x$ in time that is polynomial in the description of the game; and it is Nash-convergent if it has the property that it is guaranteed to converge to a Nash equilibrium in all games and from all starting points. 

% maybe add a sentence connecting this to learning

In view of the failure of the research community over seven decades to come up with Nash-convergent dynamics, it was conjectured in 2019 \cite{papadimitriou_game_2019} that this absence reflects an even more serious deficiency of the Nash equilibrium: perhaps there can be {\em no} dynamic that is guaranteed to converge to the Nash equilibrium in all games --- surely not a locally tractable one. Such failure would make the Nash equilibrium even less palatable as a solution concept. This line of work culminated more recently in the proof of an {\em impossibility theorem} \cite{milionis_impossibility_2023} proving the conjecture of \cite{papadimitriou_game_2019}: There is a game in which there is no dynamic whatsoever --- locally tractable or not --- that converges\footnote{The precise requirements of convergence in this theorem are somewhat technical; we go into more detail shortly.} to precisely the game's Nash equilibria. %Global convergence means that the dynamic is guided to the Nash equilibria of the game by a {\em Lyapunov function,} see Section 2. 
The decades long quest by economists for a Nash-convergent dynamic has been unproductive for concrete mathematical reasons!

However, the counterexample used in the impossibility proof of \cite{milionis_impossibility_2023} is somewhat unsatisfying, for two reasons. First, the game in question is highly {\em degenerate}. A game is degenerate if there are mixed profiles that have more pure best responses than the dimension of the game implies, due to numerical coincidences.  As this definition suggests, almost all games are non-degenerate. % --- and in view of this observation the impossibility theorem of \cite{milionis_impossibility_2023} appears to lose its bite a little.
Secondly, the convergence definition was rather strong --- the dynamic was required to converge uniformly to \emph{all} Nash equilibria in the same connected component. As the authors of \cite{milionis_impossibility_2023} point out, both assumptions are \emph{necessary} for impossibility: there is a simple dynamic (``go directly to a fixed Nash equilibrium") which converges to Nash equilibria in \emph{every game} from every starting point. The impossibility theorem proves that in some degenerate games, this convergence is only uniform on a subset of a Nash equilibrium component.
%the theorem shows thatrequires that in some degenera, and \emph{uniformly} Nash-convergent in every \emph{nondegenerate} game. 
What makes this ``go to a Nash equilibrium" toy dynamic unappealing as a model of learning, however, is its lack of \emph{local tractability} --- as we will show, computing one step of this dynamic is just as hard as computing the Nash equilibrium itself, which cannot be done in polynomial time unless $\Ptime = \PPAD$.

%Secondly, the convergence in question is \emph{uniform}, in the sense that there is a Lyapunov function showing the way (called \emph{global Nash convergence} in \cite{milionis_impossibility_2023}, see Section~\ref{sec: nash-convergent dynamics} for a detailed definition). As the authors of \cite{milionis_impossibility_2023} point out, both assumptions are \emph{necessary}: it is possible to define a simple dynamic (``go directly to a fixed Nash equilibriaum") which is Nash-convergent in \emph{every game}, and \emph{uniformly} Nash-convergent in every \emph{nondegenerate} game. What this toy dynamic lacks, however, is \emph{local tractability} --- it cannot be computed in polynomial time, unless $\Ptime = \PPAD$.

%In fact, degeneracy was {\em needed} for the Impossibility Theorem to hold, because, as the authors point out, in a nondegenerate game it is \emph{always} possible to define a Nash-convergent dynamic. In Section 2 we describe rigorously --- we believe for the first time --- two such constructions, and in fact we show here that both dynamics are {\em strongly} Nash convergent --- that is, there is a Lyapunov function showing the way (see Section~\ref{sec: preliminaries}). What these constructions fail, however, is \emph{local tractability} --- they cannot be computed in polynomial time.

%These constructionsHowever, we also note that both constructions of Nash convergent dynamics require exponential time to compute at any given point, in the worst case.  

In view of this, there seems to be only one research direction that promises to restore the spirit and intent of the impossibility theorem: {\em Change the topic back to complexity theory,} and argue that the obstacle to `natural' Nash-convergent dynamics is due to complexity. % are {\em complexity obstacles} in the way to defining, for any nondegenerate game, a locally tractable Nash convergent dynamic.
This is the goal of the present paper.

Our main finding is a sequence of theorems showing that locally tractable Nash-convergent dynamics with certain properties do not exist 
%Our main finding is that any Nash-convergent dynamics equipped with a locally tractable witness to its convergence (in the form of a \emph{Lyapunov function}) cannot be itself locally tractable 
in all games unless complexity classes collapse. First, we define two natural families of Nash-convergent dynamics that we dub Type 1 and Type 2.  Type 1 is the simplest, and embodies the example above of ``go directly to the equilibrium'' --- an incomplete version of it was already known to the authors of the impossibility theorem \cite{milionis_impossibility_2023}. Type 2 is a more sophisticated geometric construction exploiting the well-known topological fact that non-degenerate games have an {\em odd} number of Nash equilibria \cite{wilson_computing_1971}. In both cases, we assume for the sake of contradiction that we are given oracle access to an arithmetic circuit for computing the dynamic, and prove that, if indeed the circuits work as promised, a collapse of complexity classes must happen: $\Ptime = \PPAD$ for Type 1, and $\NP = \RP$ for Type 2. % The implication of these {\em black-box proofs} is that each of these types of Nash-convergent dynamics needs, in order to be computed at a point, exponential-time computation in the worst case.  Here by {\em black-box proof} we mean a proof that assumes that a locally tractable Nash-convergent dynamic exists, and concludes that a complexity class collapse must occur.
Both of these classes are what we call \emph{uniformly Nash-convergent} (see Definition~\ref{def: nash convergence}) meaning that the progress of convergence to the Nash equilibrium can be represented by a Lyapunov function. We next prove that any locally tractable uniformly Nash-convergent dynamic, if one existed, cannot possess a locally tractable \emph{Lyapunov function witness} to its convergence unless $\PPAD = \CLS$.

In the process of proving that these complexity classes collapse, we need to solve in polynomial time the following interesting computational problem related to the Nash equilibrium, which we believe has never been considered before: {\em Given a game and a line (or half line) in the strategy space, is there a Nash equilibrium on this line?}  We prove that this problem can be solved in polynomial time. In fact, the generalization in which, instead of a line, we are given an affine subspace of dimension $d$ --- and we seek the Nash equilibria on this space --- can also be solved in polynomial time, as long as the dimension is $O(1)$.  If the dimension is any fractional polynomial in the description of the game, the problem becomes NP-hard.  %We conjecture that the problem is polynomial even if the dimension of the affine space is $O({\log n\over \log \log n})$.

These findings motivate us to articulate an important conjecture: {\em If there is a Nash-convergent dynamic that can be computed locally in polynomial time for every game, then $\Ptime = \PPAD$.}  We believe that this is a central conjecture in the field of game dynamics and Nash learning, whose proof will be an important step towards understanding the way in which games, learning dynamics, and computation interact.

We also believe that proving this {\em Impossibility Conjecture} will require the development of new types of complexity-theoretic arguments, for the following reason. The existence of a polynomial-time computable Nash convergent dynamic does not immediately contradict the established difficulty of computing Nash equilibria, simply because the dynamic may very well take exponentially many steps to converge. A new kind of hybrid topological/complexity-theoretic argument seems to be needed --- one that treats learning algorithms as topological objects --- in order to establish that local polynomial-time computation creates some kind of complexity tension at remote points of a compact domain.

In fact, we prove a result that we believe reflects the mathematical and conceptual difficulty described above in complexity-theoretic terms. We define a proof protocol that we call {\em the proving game,} or {\bf PG}, which captures the essence of a large class of black-box proofs of the Impossibility Conjecture.  We prove that the adversary in {\bf PG} has a winning strategy, a result suggesting that no black-box proof of the impossibility conjecture for nondegenerate games is likely.  %This result suggests that {\em such proof is not possible in the black-box style} of complexity proofs of impossibility.  
Black-box proofs are a very common way to prove complexity class collapses, and, as far as we can see, the only kind used so far in this line of work. %In a  black-box proof a dynamic conjectured not to exist is accessed strictly through its outputs to inputs provided by an algorithm, and the proof establishes that this algorithm, using these outputs, can be used, for example, to solve by computations belonging in one class a problem that is complete for another class believed to strictly contain the first. 

To summarize, our main contributions are the following:

\begin{itemize}[nosep]
\item We provide two concrete Nash-convergent dynamics, and prove that both are intractable to compute locally under complexity assumptions.

\item As part of this proof we characterize the complexity of finding a Nash equilibrium of a game lying on an affine subspace.

\item We prove that any locally tractable uniformly Nash-convergent dynamic cannot possess a locally tractable Lyapunov function unless $\PPAD = \CLS$.

\item We articulate the new, complexity-theoretic Impossibility Conjecture.

\item We introduce the Proving Game and establish that the adversary has a forced win, demonstrating that the Impossibility conjecture may be difficult to prove.
\end{itemize}

\noindent The rest of this paper is organized as follows: After preliminary definitions, in Section~\ref{sec: nash-convergent dynamics} we discuss the definition of Nash-convergence, and define the two types of Nash-convergent dynamics sketched in the introduction.  In Section~\ref{sec: hardness} we prove complexity collapse results for them and for any Nash-convergent dynamic equipped with a Lyapunov function.  In Section~\ref{sec: impossibility for uniform} we articulate the Impossibility Conjecture, and we prove a result suggesting that, in the black-box proof style of complexity proofs --- the main tool available for such quests ---, the Impossibility Conjecture cannot exist.  Finally, in the last section we discuss the implications of these results and contemplate the way forward.

\section{Preliminaries} \label{sec: preliminaries}
\subsection*{Games}
A game $G$ has some number $p$ of players and a set of strategies $S_i$ for each of the $p$ players. A {\em strategy profile} is a choice of a strategy for each of the $p$ players, whereas a {\em mixed profile} is a choice of one distribution $[p_j: j\in S_i]$ for each player $i$.  The game $G$ is specified by the {\em utilities}, a set of $p$ $p$-dimensional tensors $u^i, i=1\ldots,p$ of rational numbers where $u^i_{j_1}, \ldots, u_{j_p}$ is the benefit to player $i$ if the $p$ players play strategies $j_1,\ldots,j_p$ respectively.  The utilities of a mixed profile are the expectations of the utilities under the independent probabilities of the profile, one for each player.  The set of all mixed profiles of the game is denoted by $X_G$, or simply $X$ when the game is implicit.  The {\em size of game $G$,} denoted $|G|$, is the sum of the lengths of the representations of the utilities  of the players, assumed to be rationals in $[0,1]$.  Up to a polynomial and assuming reasonably small accuracy of utilities and a finite number of players, the total number of strategies of a game, denoted by $n$, is a useful surrogate of its size.  

\medskip{\bf Degeneracy} is simplest to define in two-player games \cite{von_stengel_computing_2002}: A bimatrix game is nondegenerate if there is no mixed strategy with support $k$ that has more than $k$ pure best responses.  It is more subtle in the case of three or more players, for which no consensus exists on the precise definition of degeneracy --- see \cite{von_stengel_computing_2002,mckelvey1996computation} also for various treatments of the subject.  Since the Nash learning problem is already challenging and important for the intractable case of two players, we shall focus our discussion on two-player games. However, we shall also be pointing out that several positive results hold for games with more players. Telling whether a two-player game is degenerate is known to be NP-hard, a property it inherits from the degeneracy of linear programming \cite{du2013complexity,murty1987some}, but this interferes surprisingly little with our narrative.% --- we do point out that finding a Nash equilibrium remains PPAD-complete even if the game is nondegenerate.

A {\em Nash equilibrium} is a mixed profile such that there is no player and no strategy of this player such that the player improves their utility by changing the probability of the strategy in the mixed profile. The Nash equilibrium is of paramount importance in Game Theory, and has been considered for 75 years as the default ``meaning of the game,'' the place where the action is destined to end up. However, as recounted in the introduction, intractability and non-uniqueness have shaken the field's faith in the Nash equilibrium. A common counterargument of its adherents is that pairs or groups of players, or institutions and societies, will eventually and spontaneously end up at a Nash equilibrium --- despite the adverse evidence.  This brings in the important subject of {\em dynamical systems}.

\subsection* {Dynamical Systems}
A {\em continuous-time dynamical system,} or a {\em continuous-time dynamic,} on a compact real set $X$ is a continuous function $\phi(x,t): X\times \real_+\mapsto X$ satisfying $\phi(\phi(x,t),t'))=\phi(x, t+t')$ denoting the point at which $x$ will be after time $t$.  Such a dynamic is typically represented  in terms of its time derivative $\dot\phi(x) := \frac{\mathrm{d}}{\mathrm{d}t} \phi(x)$.  %, and in particular through an {\em arithmetic circuit} computing $ \phi$ on input $x$ (see, for example, \cite{CLS} for an appropriate definition of arithmetic circuits that avoids the creation of exponentially large results). 
Also of great interest are {\em  discrete-time dynamics,} which are continuous maps $\phi(x,t): X\times \mathbb{Z}_+\mapsto X$ satisfying the same property $\phi(\phi(x,t),t'))=\phi(x, t+t')$. A discrete-time dynamic can be defined as repeated function application of a single continuous function $f: X \to X$, with $\phi(x,t) := f^{t}(x) = f(f(\dots(f(x))\dots)$. While these two definitions differ in important ways, the two theories of continuous- and discrete-time dynamical systems are surprisingly parallel --- for example, the impossibility result mentioned in the introduction, as well as the Fundamental Theorem of Dynamical System to be explained soon, hold for both. {\em In this paper, we shall work mostly with discrete-time dynamics;} however we occasionally use continuous time dynamics, and an exposition with the continuous-time variety would also be possible.

The purpose of the mathematical field of dynamical systems is to predict --- and, implicitly, compute --- the limit behavior of the system as time goes to infinity, when the system is started at some initial point in $X$. The simplest possible limit behavior would be for the dynamical system to converge to a fixed point, or, a bit more generally, a {\em cycle}. In fact, the Poincar\'e-Bendixson Theorem from the early 1900s states that, for two-dimensional dynamical systems (e.g., for $2\times 2$ games), these are the only possible limit behaviors (intuitively, because planarity constrains the system's behavior).  In three or more dimensions, however, far more complex limit behaviors prevail, including the phenomenon of chaos. Researchers strove for over a century to restore the simplicity of Poincar\'e-Bendixson Theorem to higher dimensions, and these efforts culminated in the 1980s with the Fundamental Theorem of Dynamical Systems due to Charles Conley \cite{conley_isolated_1978}.  It states that all dynamical systems will eventually cycle, provided one redefines cycling: Call two points {\em chain recurrent} if, for every $\epsilon$ there is an $n$ such that the system can go from $a$ to $b$ and back by at most $n$ steps of the dynamical system, intermingled with arbitrary $\epsilon$ jumps.  That is, two points are chain recurrent if one can reach the other through the computation of the system's steps, {\em if} an agent can manipulate appropriately the round-off errors.  The definition for continuous-time systems replaces the steps with trajectory segments of unit time length. Chain recurrence is an equivalence relation, and Conley proved that the system will {\em globally converge} to one of the equivalence classes, called {\em chain recurrent components.}  This means that it will not only converge, but will do so guided by a {\em Lyapunov function,} a function $L$ mapping $X$ to $\real_+$ such that, $L$ is constant on equivalence classes under chain recurrence, and for all other points $x\in X$, $L(x) > L(\phi(x))$. Note that Conley's Theorem nearly restores the Poincar\'e-Bendixson Theorem, since it establishes that dynamical systems always converge to objects that are arbitrarily close to cycles. %Global convergence is stronger, and more desirable, than plain convergence because, intuitively, the latter --- by the definition of chain recurrence --- 

Dynamical systems are important in game theory. Starting even before Nash \cite{brown1949some}, dozens of dynamical systems have been proposed by game theorists and economists, often with the explicit goal to discover natural economic behaviors that eventually lead to the Nash equilibrium. Unfortunately, these systems fail to converge to Nash equilibria in general games, with convergence results typically limited to classes with special structure such as potential games \cite{monderer_potential_1996}, two-player zero-sum games or two-player $2\times n$ games \cite{berger_fictitious_2005}. %The Impossibility Theorem of \cite{milionis_impossibility_2023} states that there is a highly degenerate two-player $3\times 3$ game for which there is no dynamical system that \emph{globally Nash-convergent}. This convergence requirement is fairly strong (see Definition~\ref{def: nash convergence}), and the goal in the present paper is to use Complexity Theory to extend the reach of the Impossibility Theorem to nondegenerate games and to less restrictive definitions of Nash convergence.  

\subsection*{Nash Convergence}
What does it mean for a game dynamic $\phi$ to be Nash convergent? This question is critical to this paper, and there are a variety of definitions one might consider. Here we lay out several candidate definitions, and discuss how they relate to each other and to the Impossibility Theorem of \cite{milionis_impossibility_2023}.

\begin{definition} \label{def: nash convergence}
We describe several definitions, from most basic to most demanding.
    \begin{enumerate} [label=(\Alph*)]
\item \emph{Nash-convergence}: for all $x\in X_G$, $\lim_{t\rightarrow\infty} \phi(x,t)$ exists and is a Nash equilibrium. 

\item {\em Uniform Nash convergence} is (A) with a {\em Lyapunov function}, $L:X_G\mapsto [0,1],$ such that, if $x$ is not a Nash equilibrium, $L(\phi(x,t))< L(x)$ for all $t>0$.   %Any uniformly converging dynamic can be made to also be Nash stationary, as above.

\item[(A$^+$)/(B$^+$)] is (A) or (B) with the addition of {\em Nash stationarity}, the requirement that all Nash equilibria are fixed points of the dynamic\footnote{Note that (A) already implies that every fixpoint is a Nash equilibrium --- Nash stationarity requires the converse.}.

\item Finally, and most strongly is what is called \emph{global Nash-convergence} in \cite{milionis_impossibility_2023}; this is (B$^+$) with the additional requirement that, if $N$ and $N'$ are Nash equilibria in the same connected component of the Nash equilibrium set, then $L(N)=L(N')$.
\end{enumerate}
\end{definition}

The first definition (A) is the most obvious definition of Nash-convergence. However, this `vanilla' Nash-convergence does allow for certain counterintuitive behaviors, such as convergence to {\em homoclinic and heteroclinic orbits}\footnote{A heteroclinic orbit is a cycle of fixpoints joined by connecting trajectories (a homoclinic orbit is a trajcetory from a fixpoint to itself). Even if each point converges to a Nash equilibrium, the overall behavior is cyclic and no individual Nash equilibrium is stable.} of Nash equilibria. Uniform Nash-convergence (B) rules out these possibilities, and guarantees that at least some Nash equilibria are (Lyapunov) stable under the dynamic (nearly points do not move away from the equilibrium). Nash stationarity (types (A$^+$) and (B$^+$)) seems at first to impose some additional constraints, but in fact does not. In actuality, it is straightforward to show that any dynamic satisfying (A) or (B) can be readily converted into one satisfying stationarity by a change of speed --- one multiplies the gradient of $\phi$ by a nonnegative continuous function that is zero at exactly the Nash equilibria of the game.  There are many standard such functions, including the Brown-von Neumann-Nash function introduced later in this paper. Because this change of speed can be done in a locally tractable way, Nash stationarity does not affect the local tractability of $\phi$.

The Impossibility Theorem of \cite{milionis_impossibility_2023} applies to the most restrictive definition (C); they show that, for a particular $3\times 3$ degenerate game, there is no globally Nash-convergent dynamic. On the other hand, (locally intractable) dynamics do exist for the convergence definitions (A) and (B) (optionally with Nash-stationarity) --- an example is the Type 1 dynamics we define in the next section. The goal of this paper is to use Complexity Theory to extend the reach of the Impossibility Theorem in two directions: (1) nondegenerate games, and (2) dynamics of kinds (A) and (B) (recall that the additional requirement of Nash stationarity does not make the problem any harder).

%\subsection{Nash convergent dynamics}

\section{Two families of Nash-Convergent Dynamics} \label{sec: nash-convergent dynamics}

In this section we discuss two families of Nash-convergent dynamics which exist in any non-degenerate game. In the next section we show that these dynamics are intractable to compute, and we conjecture that this is a general property of Nash-convergent dynamics.

\paragraph{Type 1.}  Let $G$ be a game, and let $Z$ be its set of Nash equilibria. %If the game is nondegenerate, $Z$ is finite. %Pick any Since the game $G$ is assumed nondegenerate, it follows that its Nash equilibria comprise a set of distinct points, and we call this set $Z$\footnote{This property, as well as the ``oddness'' property needed in the next subsection on Type 2, also hold for all conceptions of degeneracy in games with more than two players.}. 
A dynamic of Type 1 is very simple: we pick one $z\in Z$ from these equilibria, and define a dynamic $\phi_{G,z}$ in which each mixed profile $x\in X$ moves along a straight line in the direction of $z$ with constant speed. This $\phi$ is continuous and uniformly Nash-convergent, with the distance from $z$ acting as a Lyapunov function.

\begin{figure}
    \centering
    \begin{tikzpicture}[
    inward arrows/.pic = {
        \node (c) at (0,0) [draw, circle,fill,inner sep=2pt] {};
        \foreach \angle in {0, 90, 180, 270} {
            \draw[->, shorten > = 0.1cm] (\angle:.6cm) -- (c);
        }
    }]

    % Outside boundary
    \draw (0,0) to (-2,2) to (0,4) to (4,4) to (6,2) to (4,0) to cycle;
    \draw[very thick] (.5,0) to (.5,4);
    \draw[very thick] (3.5,0) to (3.5,4);

    %\node[circle,fill,inner sep=2pt] (1) at (-.5,1.5) {};
    \draw (-0.7,1.7) pic [rotate=-10] {inward arrows};
    \node[circle,fill,inner sep=2pt] (2) at (.5,2.5) {};
    \draw (2,2) pic [rotate=10] {inward arrows};
    %\node[circle,fill,inner sep=2pt] (3) at (2,2) {};
    \node[circle,fill,inner sep=2pt] (4) at (3.5,1.5) {};
    %\node[circle,fill,inner sep=2pt] (5) at (4.5,2) {};
    \draw (4.7,2) pic [rotate=-5] {inward arrows};

    \node at (-1.05,1.4) {$z_1$};
    \node at (0.1,2.1) {$z_2$};
    \node at (1.7,1.6) {$z_3$};
    \node at (3.1,1.1) {$z_4$};
    \node at (4.4,1.7) {$z_5$};

    \draw[->] (.4,3.2) to [in = 0, out = -90] (0.1,2.6);
    \draw[->] (.4,1.8) to [in = 0, out = 90] (0.1,2.4);
    \draw[->] (.6,3.2) to [in = 180, out = -90] (0.9,2.6);
    \draw[->] (.6,1.8) to [in = 180, out = 90] (0.9,2.4);

    \draw[->] (3.4,2.2) to [in = 0, out = -90] (3.1,1.6);
    \draw[->] (3.4,0.8) to [in = 0, out = 90] (3.1,1.4);
    \draw[->] (3.6,2.2) to [in = 180, out = -90] (3.9,1.6);
    \draw[->] (3.6,0.8) to [in = 180, out = 90] (3.9,1.4);

    \node at (-1.3,-.2) {$P_1$};
    \node at (0.5,-.2) {$P_2$};
    \node at (2,-.2) {$P_3$};
    \node at (3.5,-.2) {$P_4$};
    \node at (4.7,-.2) {$P_5$};

\end{tikzpicture}
    \caption{An illustration of the Type 2 dynamic. The space is divided into an odd number of polytopes $P_i$, each containing a Nash equilibrium $z_i$. The odd-numbered polytopes have dimension $d$, while the even-numbered polytopes have dimension $d-1$.}
    \label{fig:type 2 dynamic}
\end{figure}
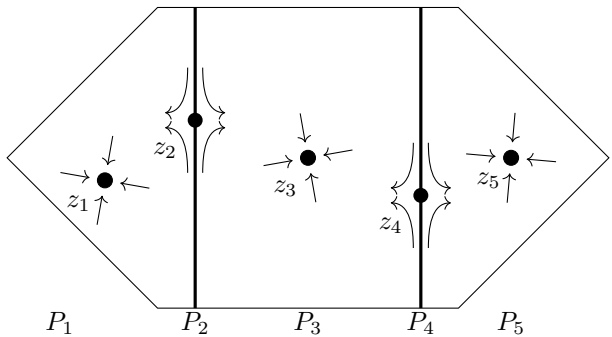

\paragraph{Type 2.}   We next present a different type of Nash-convergent dynamic guaranteed to exist in any non-degenerate game with any number of players. %This is a folk construction, which, however, has never appeared in print.  
To describe the dynamic, let $z_1,z_2,\ldots, z_{2k+1}$ be the  set of Nash equilibria of the game --- %off by Harsanyi's ``oddity theorem.'' T
that the total number of Nash equilibria is odd in a nondegenerate game is known as Wilson's oddness theorem \cite{wilson_computing_1971}.  We next identify a line $\ell=\{y + \lambda\cdot w: \lambda\in R\}$ where $y\in X$ and $w\in R^d - \{\bf{0}\}$.  Line $\ell$ is chosen so that all projections of the Nash equilibria on $\ell$ are {\em distinct}.  It is not hard to see that such an $\ell$ always exists.  For example, a random line satisfies this property almost certainly, and a simple deterministic construction is also possible: start from any line and perturb it appropriately if two or more projections of Nash equilibria on it coincide.

Now we divide $X$ into $2k+1$ polytopes, as follows (see Figure~\ref{fig:type 2 dynamic} for a schematic illustration). We define the hyperplanes $H_i: i = 1,\ldots, k$ as the hyperplanes normal to $\ell$ at the points $z_2, z_4,\ldots, z_{2k}$.  This divides $X$ into $k+1$ $d$-dimensional polytopes $P_1,P_3\ldots, P_{2k+1}$, and also defines $k$ $d-1$-dimensional polytopes $P_2, P_4, \ldots P_{2k}$ that are intersections of two consecutive odd-indexed polytopes. Notice that, for all $i$, the closed polytope $P_i$ contains only equilibrium $z_i$ if $i$ is even, and equilibria $z_{i-1}, z_i, z_{i+1}$ if $i$ is odd (with the obvious adjustments for the boundary cases $i=1,2k+1$).

We next define the dynamic $\phi_{G,\ell}$. Inside each even-indexed polytope $P_{2i}$, $: \phi_{G,\ell}(x) = x+ \alpha\cdot {(z_{2i}-x)}$,
%\over |z_{2i}-x|},$ 
where $\alpha < 1$.  That is, in even-indexed polytopes the dynamic moves towards the unique equilibrium.  

In the odd-indexed polytopes, the dynamic also moves towards the corresponding, odd-numbered equilibrium.  However, to ensure continuity with the neighboring even polytopes, the dynamic is perturbed so that continuity is restored: it starts equal to the dynamic of the even polytope when it is close to the boundary, and gradually it moves more and more in the direction of its own Nash equilibrium.

For more details, we divide every odd-indexed polytope $P_{2i+1}$, %with the exception of the extreme ones $P_1$ and $P_{2k+1}$, 
into two closed {\em chambers} $P^-_{2i+1}$ and $P^+_{2i+1}$ by the hyperplane normal to $\ell$ through $x_{2i+1}$, with $P^-_{2i+1}$ being the part containing $z_{2i}$.
%In even-numbered polyopes, containing Nash equilibrium $z$, as well as in the chambers $P_1^-$ and $P_{2k+1}^+$, the dynamic is simply moving towards the Nash equilibrium in the chamber:
%$\phi(x) = \alpha\cdot(z -x),$ for some $\alpha < 1$. 
The width of chamber $P^-_{2i+1}$ --- the distance between its two bounding hyperplanes normal to $\ell$ --- is denoted $w^-_{2i+1}$.  To define $\phi$ within chamber $P^-_{2i+1}$, let $d^-(x)$ be the distance of $x$ from the polytope $P_{2i}$, and let $x^-$ be the closest point of $x$ on $P_{2i}$. %, while $d^+(x)$ is the distance of $x$ from the intersection of the hyperplane normal to $\ell$ at $z_{2i+1}$, and $x^+$ is the closest point.
%recall that the distance of a point outside a polytope and the polytope is well defined and easy to compute.  
The dynamic inside the chamber is this: 

$$\phi_{G,\ell}(x) = {\frac{w^-_{2i+1}-d^-(x)}{w^-_{2i+1}}\cdot (x+(\phi(x^-)-x^-)} + \frac{d^-(x)}{w^-_{2i+1}}\cdot (x+ \frac{(z_{2i+1} - x)}{k}),$$
for some $k\gg 1$.

The dynamic at the $+$ chambers is defined completely analogously.

%In English, in even-indexed polytopes the dynamic moves towards the unique equilibrium, as in Type 1, while in odd-indexed polytopes the dynamic is an interpolation between moving towards the unique equilibrium of the polytope and the dynamic at the next even-indexed polytope, to ensure the continuity of the dynamic (see Figure ?? for a schematic illustration).  
Finally, for the two extreme odd polytopes, the dynamic is defined by a similar interpolation between (a) motion to the Nash equilibrium and (b) the dynamic in the adjacent even polytope, scaled by distance to the latter polytope.

Notice that, as expected, it is not clear how to compute this dynamic at any point if the Nash equilibria are not given.  It is easy to check that the dynamic is well-defined for all $G$ and appropriate $\ell$, that it is continuous by construction, and that its only fixpoints are the game's Nash equilibria:

\begin{proposition}
The dynamic $\phi_{G,\ell}$ defined above is uniformly Nash-convergent and Nash stationary (type (B$^+$) in Def.~\ref{def: nash convergence}).
\end{proposition}

The Lyapunov function for even polytopes is the distance to the equilibrium.  For odd polytopes it is only a little harder to define.

\section{Hardness Results for Nash-Convergence} \label{sec: hardness}
In this section we explain why locally tractable Nash-convergent dynamics do not exist for any dynamics of either of these types. We prove this by contradiction. First, we assume that such a dynamic exists. Then we use this hypothesized dynamic as a {\em black box} to solve a difficult problem, and thus show that two complexity classes that are not thought to be equal are equal. For type 1, we use $\Ptime$ vs $\PPAD$, and for Type 2 we use $\NP$ and $\RP$.  We conclude that there is a complexity obstacle to the existence of such a dynamic.  This style of complexity-theoretic argument is known as a {\em black-box proof\/} because it proceeds with no consideration to the precise details of how the dynamic is implemented. 
%Because we use uniform Nash-convergence, each of these dynamics assume the game is non-degenerate --- in each black-box proof, either we show a difficult problem can be solved or point out that the game is in fact degenerate. 
Finally, in Section~\ref{sec: impossibility for uniform} we prove a more general theorem showing that \emph{any} uniformly Nash-convergent dynamic (not necessarily Type 1 or 2) \emph{equipped with a locally tractable Lyapunov function} cannot be locally tractable unless $\CLS = \PPAD$.

For Type 1 dynamics, the argument requires the solution of a basic computational problem ({\em ``is there a Nash equilibrium on a given line?''}) addressed in the next subsection.  For Type 2, the argument is somewhat nontrivial and relies on the same algorithm. %For the third type, the proof is a full-blown reduction, albeit a straightforward one, in which we show that any dynamic of the required kind would establish the collapse of two complexity classes that are not believed to coincide, CLS and PPAD.  

%Finally, Type 4 seems even harder to prove intractable: {\em we have been unable} to find a black-box proof for this, most general, case. In the next section we show that there is a mathematical reason why a black-box proof of this impossibility result is unlikely to exist, and thus a novel kind of complexity argument may be needed to settle this most general impossibility conjecture.

%\textcolor{blue}{\bf At some points in the section, a clause ``or present evidence of degeneracy" should appear.  Where?}

\subsection{Nash Equilibria in an Affine Subspace}
The proofs of this section require the efficient solution of the following problem:  Given a line $x=A+B\cdot\lambda$ for vectors $A$ and $B$ in the mixed strategy space $X$, is there a Nash equilibrium on this line? More generally, given an affine subspace, can we find a Nash equilibrium lying on this subspace?  %In this subsection we explore two approaches to this interesting problem.   
Define {\sc Nash on an affine space} to be the following computational problem:  Given a game $G$ and an affine space $\mathcal A$ of dimension $d$ defined by the equations
$x_i=A_i + \sum_{j=1}^d B_{ij}\cdot \lambda_j,$
one for each strategy (dimension) $i$ of the game, find all Nash equilibria that lie in $\mathcal A$ or determine that none exist.  In this subsection we prove the following:

\begin{theorem} The {\sc Nash on an affine space} problem can be solved in polynomial time if $d$ is a constant, and is NP-hard if $d=\Omega(|G|^\delta))$ for any $\delta>0$, where by $|G|$ we denote the size of the representation of game $G$.  
\end{theorem}

%Notice that Mehta's proof above {\em does} involve the base case of the elimination of existential quantifiers, as it works by replacing the existentially quantified matrix inequality with the disjunction of unquantified inequalities.  

\iffalse
The utility function for $p$ players on the affine subspace is
$U(\lambda_1,\ldots,\lambda_d) = \sum_{i_1,\ldots,i_p} u_{i_1,\ldots,i_p}\prod_{k=1}^p(A_{i_k}+\sum_{j=1}^d B_{i_k,j}\lambda_j)$, and 
any Nash equilibrium of $G$ on the subspace is a KKT stationary point of the utility.  All these KKT stationary points can be computed through Tarksi elimination: It is known \cite{Basu} that, given $N$ polynomials in $V$ variables of maximum degree $D$ and a Boolean function on the signs of the polynomials that is existentially quantified over the variables, it can be decided in time $O((NDV)^{O(D)})$ whether the quantified Boolean formula is valid. Furthermore, the total number of solutions is $O(D^{O(D)})$.  Specifying the sought points can be expressed as a Boolean function, and the ponts can be further isolated Notice that, for the range of $D,N,V$ pertinent to the polynomials considered here, the bounds of this result on both time and number of solutions are polynomial in $|G|$ when $d=O(\frac{\log|G|}{\log\log |G|})$
\fi

\begin{proof}
Consider the intersection of the space of mixed strategies $X$ with the given $d$-dimensional affine space $A$.  It is a polytope of dimension $d$ and $O(n)$ facets, where by $n=O(|G|)$ we denote the total number of strategies of $G$.  By the dual of the upper bound theorem (see \cite{alon1985simple}) this polytope has $n^{O(d)}$ faces, and each of these faces is the intersection of $\mathcal A$ with a face of $X$. In addition, it is known that these faces can be enumerated with polynomial delay per output, see \cite{fukuda2020polyhedral}. Therefore, the Nash equilibria in $\mathcal A$ can be found by enumerating these faces in time $n^{O(d)}$ and, for each face --- corresponding to a possible support of mixed strategies in the game --- we test whether (a) the interior of the face contains a Nash equilibrium, achievable by solving a system of linear equations; and (b) whether this equilibrium --- guaranteed to be unique --- also belongs in the affine space $\mathcal A$.  The total time for this test is $n^{O(d)}$.  Since $d$ is a constant, this algorithm runs in polynomial time in $n$, the description of the game.

%\paragraph{NP-hardness.}
Now for the NP-hardness part of the Theorem:  When $d$ attains its largest possible value, one less than the dimension of $X$, we are seeking a Nash equilibrium of $G$ lying on a hyperplane --- that is, satisfying a linear equation.  This is known to be NP-hard, since it is NP-hard to test, even for two-player games, whether there is a Nash equilibrium in which a particular component of the mixed strategy of a particular player is a given constant $c>0$ \cite{conitzer2008new}.   To prove the result for $d =\Theta(|G|^\delta)$, we pad the game with $|G|^\frac{1}{\delta}$ irrelevant strategies.  Also, for intermediate values of $d$, it can be shown that the problem can be solved in subexponential time, and is NP-hard via subexponential reductions.  
\end{proof}

\paragraph{Remark.} When $d=1$ and the number of players is two, there is a clever {\em algebraic} algorithm for this problem due to Ruta Mehta (private communication, 2025):  As is well known, the Nash equilibria of any two-player game $G$ coincide with the symmetric Nash equilibria of the {\em symmetrization of $G$,} a game of the form $(A, A^T)$.  So, it suffices to consider the corresponding line in this space and find the symmetric Nash equilibria that it contains, if any.  These equilibria are the solutions of the following system of degree-two inequalities: $(Ax)_i \le x^T A x, i=1,\ldots,n$.
Since the sought equilibria are required to lie on a line parametrized by $\lambda$, we can reduce this system to a system of quadratic inequalities involving the single variable $\lambda$. By sorting the roots of these quadratic polynomials in $\lambda$, we can solve the system of inequalities and identify all Nash equilibria on the line, if any.  

We conjecture that Mehta's algorithm can be generalized to more than two players through Tarski's elimination of quantifiers for the theory of reals \cite{basu1996combinatorial}, providing an alternative proof of a stronger version of the positive direction of the theorem, in which the complexity is polynomial even if $d=O(\frac{\log |G|}{\log\log |G|})$, where $|G|$ is the size of the game description.  Notice that Mehta's maneuver of solving the quadratic equations and sorting the roots is essentially the elimination of a single existential quantifier.

\subsection{Type 1 Dynamics} 
We now turn to applying the affine space result  with $d=1$ --- that is, telling whether a Nash equilibrium lies on a line --- in order to prove complexity lower bounds for Nash convergent dynamics in games of these two types, starting with the easiest case of Type 1. 
\begin{theorem}
If a locally tractable Nash-convergent dynamic of Type 1 exists for all games, then $\Ptime = \PPAD$.
\end{theorem}
\begin{proof}
Let $G$ be a game, and let $\phi$ be the assumed dynamic.  It is easy to see that we can use the circuit computing $\phi$, together with the algorithm for finding a Nash equilibrium on a line, to find a Nash equilibrium of $G$.  Since this PPAD-complete to do for arbitrary games, this completes the proof.  
\end{proof}

\subsection{Type 2 Dynamics}
\begin{theorem}
If a locally tractable Nash-convergent dynamic of Type 2 exists for all nondegenerate games, then $\NP = \RP$.
\end{theorem}
\begin{proof}
We shall prove that, given oracle access to the hypothesized dynamic $\phi$, one can decide with arbitrarily high probability whether the game $G$ has a unique Nash equilibrium, or more than one.  Since the problem is NP-hard \cite{conitzer2008new}, this would complete the proof\footnote{One may ask at this point: ``Does the problem of telling whether a game $G$ has more than one Nash equilibria remain NP-complete even for nondegenerate games?'' Because degeneracy is destroyed by arbitrarily small perturbations, we can modify the original NP-hardness reduction of \cite{gilboa_nash_1989} to this setting. This argument reduces the {\sc Clique} problem to the problem of testing for a Nash with given utility; because the utility bounds in this reduction are not tight, one can construct a \emph{randomized} form of the reduction by perturbing the utilites by a very small amount to make the instance nondegenerate. Interestingly, making this constructive would require the solution of an NP-hard problem. More broadly though, we don't need NP-completeness in nondegenerate games for our argument, because our type 2 construction applies to any game with finite equilibria, a larger class.}.  The test is simply the following:  Choose a point in $X$ at random, compute the direction of $\phi$ on this point, and test if there is a Nash equilibrium on this line.  If so, then the game almost certainly has a unique Nash equilibrium.  If not, the game has multiple equilibria. %\footnote{Notice that here we could instead perform a more demanding test, choose $k$ points and require that they all point to the same Nash equilibrium; however, this stronger requirement is not necessary for the argument.}.

The key properties of the dynamic of Type 2 that ensure the correctness of the algorithm are that, (a) if the game has one Nash equilibrium, $\phi$ heads directly towards it from all points in $X$, since there is only one polytope, $X$; and (b) if there are multiple equilibria, then in the interiors of all chambers $\phi$ is perturbed slightly from pointing to the Nash equilibrium of the chamber, in order to be reconciled with the dynamic at the next even polytope.  As a result, there are only two kinds of points in $X$ that can possibly point to equilibria: 
\begin{itemize}
\item Those that happen to lie on the boundary between two chambers, a set of measure zero; and

\item  For each chamber, and each orbit in the chamber, a set of points in a polytope that happen to point to a Nash equilibrium {\em in another polytope.} It is not hard to argue that this set of points is at most countable for each orbit in the chamber and each target Nash equilibrium, and thus also of measure zero.  
\end{itemize}

To add some detail to the last paragraph, consider the orbit of a generic point $x$ in an even polytope $P_e$ both in $P_e$ and in an adjacent odd polytope $P_o$. In $P_e$ the path is $x+t\cdot (N_e-x), 0\leq t\leq 1$, while in $P_o$ the motion may be $x+t^2\cdot (N_e-x), 0\leq t\leq 1$, where $N_o$ is the equilibrium on $P_o$ and $N_e$ is the equilibrium on $P_e$. Consider another equilibrium $N$, and the possibility that the motion's direction at some point $p$ of the latter orbit goes through $N$.  For this to happen, $N$ must be on the same 2-dimensional affine subspace defined by $x,N_o,N_e$, and even then, because of convexity, at most one of the uncountable points $p$ on the orbit would qualify.  Clearly, a vanishing fraction of the points of the curve point to $N$, and if we sum this vanishing fraction over the finitely many Nash equilibria we conclude that the measure of the points of $X_G$ from which the dynamic points to a Nash equilibrium is zero.  

This completes the proof.

\end{proof}

%Note that the argument can be easily formulated in terms of discrete computation and discrete probabilities.  Also, 

We suspect that, by a more specific construction of the dynamic, the complexity implication of this theorem can be strengthened to $\NP = \Ptime$.

\subsection{The Impossibility Theorem for Uniform Nash Convergence} \label{sec: impossibility for uniform}

By Definition~\ref{def: nash convergence}, any uniformly Nash-convergent dynamic must possess a Lyapunov function which is decreasing on non-Nash points, and constant on Nash equilibria that are fixpoints. From the previous section, we know that such dynamics exist in all games, but the ones we know happen to be intractable. We conjecture that no such dynamics can be locally tractable, and in this section we prove a slightly weaker form of this: if a uniformly Nash-convergent dynamic came equipped with a locally tractable Lyapunov function, then it cannot be locally tractable unless $\PPAD = \CLS$. In other words, a locally tractable uniformly Nash-convergent dynamic, if one existed, cannot provide a \emph{locally tractable witness to its uniform convergence}, in the form of a Lyapunov function.

We next make this intuition precise. We are given a polynomial $\pi$ and two fixed arithmetic circuits $f$ and $L$ which will compute the dynamic $\phi(x,t) := f^t(x)$ and its Lyapunov function $L(x)$ respectively. That is, when presented with a game $G$ and a strategy profile $x$ of $G$ with rational entries with $\pi(|G|)$ bits, the circuits will compute another rational profile $f(x)$ and the value of the Lyapunov function $L(x)\in [2^{\pi(|G|)}]$ with the promise that, for all non-Nash $x$, $L(f(x))>L(x)$.  It is further promised that any sequence of the form [$x,f(x),f(f(x)),\ldots]$ will end up after finitely many steps at a Nash equilibrium of $G$. We now show that a Nash-convergent dynamic cannot be locally tractable and possess a locally tractable Lyapunov function unless $\PPAD = \CLS$.

\begin{theorem} \label{CLS theorem}
Suppose that a pair $f$ and $L$ of arithmetic circuits exist where $f$ defines a locally tractable Nash-convergent dynamic $\phi(x,t) := f^t(x)$ for all games and $L$ defines a locally tractable Lyapunov function for $\phi$. Then $\PPAD = \CLS$. 
\end{theorem}
%To begin, we shall define precisely the terms of the theorem, starting with what it is meant by ``a locally tractable strongly Nash-convergent dynamic.''  This means that there is a  %For technical reasons, we shall also assume that, if $x$ is a Nash equilibrium of $G$, $\phi(x)$ is $x$ with an extra check mark symbol added.

First, recall that TFNP is the subclass of FNP consisting of all function problems that are total, in that every input has an output.  The problem of factoring, and the Nash equilibrium problem are examples.  The complexity class PPAD is the class of all problems in TFNP reducible to finding a Nash equilibrium, while the class PLS consists of all problems in TFNP that reduce to the following problem, called {\sc Iter}: We are given a Boolean circuit $C$ computing a function $C:[2^N]\mapsto [2^N]$ such that $C(0)>0$ and $C(i)\geq i$ for all $i$.  We are asked to find a $y\in [2^N]$ such that $C(y)>y$ and $C(C(y))=C(y)$.  It should be clear that this is a problem in TFNP.

Finally, CLS has been recently proved \cite{fearnley2021cls} to be exactly the intersection of PPAD and PLS (it was previously known to be a subset of this intersection).  CLS  is broadly believed to be different from both PLS and PPAD.

We now come to the proof.
\begin{proof}  
Under the assumption of the theorem, we shall reduce the problem of finding a Nash equilibrium of a given two-player game to {\sc Iter}.  
%Since the Nash equilibrium problem is PPAD-complete even for nondegenerate games --- the original reduction \cite{chen2006settling} is robust to small perturbations of the parameters --- 
It will follow that PPAD is a subset of PLS, and thus equal to CLS (the intersection of PLS and PPAD).

Given $G$, we construct an {\sc Iter} circuit $C:[2^{2\pi(|G|)+1}]\mapsto [2^{2\pi(|G|)+1}]$ as follows. We will encode mixed profiles and their Lyapunov function values as pairs $y_1,y_2$, where $y_2$ is a representation of the profile and $y_1$ represents $L(y_2)$. If the string $y\in [2^{2\pi(|G|+1}]$ is not of the correct form (that is, $y_2$ is a valid representation of a mixed profile and $L(y_2) = y_1$) then we set $C(y) = y$. Otherwise we define $C(y_1,y_2)=(L(f(y_2)),f(y_2))$, which is the successor profile $f(y_2)$ under the Nash-convergent dynamical system, accompanied by the Lyapunov function of $f(y_2)$. If $y_2$ is a Nash equilibrium, then $f(y_2) = y_2$ and so $C(y) = y$. Finally, we choose a mixed strategy $x_0$ of $G$ and set $C(0)$ to be $(L(x_0),x_0)$.

We claim that this is a valid reduction from the Nash equilibrium problem to {\sc Iter}, as needed.  All we need to show is that from any $y\in[2^{2\pi(|G|)}+1]$ with $C(y)>y$ and $C(C(y))=C(y)$ --- recall that we are guaranteed such an output --- solves the Nash equilibrium problem. That is $C(y)$ is a fixed point of $C$, but has a distinct predecessor $y$. If the string $y$ were an illegitimate encoding of a configuration, then we defined $C(y) = y$, so we conclude that $y = y_1,y_2$ must be a valid encoding of a mixed profile $y_2$ and its Lyapunov function value $y_1$. Then $C(y_1,y_2) = (L(z),z)$, where $z := f(y_2)$ is a mixed profile, and further $C(C(y)) = (L(f(z)), f(z)) = C(y) = (L(z),z)$, giving $f(z) = z$. Hence $z$ is a valid mixed profile that is a fixed point under our Nash-convergent dynamic $f$, which implies that it is a Nash equilibrium.

%Otherwise, suppose that $C(y)$ is not a Nash equilibrium of $G$ preceded by its Lyapunov function.  Then it must be an illegitimate encoding of a configuration, and yet it was produced from $y$ by $C$, which is impossible.
%. The point $y$ is a valid representation of a mixed strategy profile, as $C(y)\neq y$, and together this implies that the dynamic malfunctions on $G$, and hence the game is degenerate, with $y$ as a certificate.
\end{proof}

\section{The Impossibility Conjecture for General Dynamics}

We finally articulate the overarching open research problem proposed in this paper: 
\begin{conjecture}[The Impossibility Conjecture]
    If a locally tractable Nash-convergent dynamic exists for every game, then $\Ptime = \PPAD$.
\end{conjecture}

Naturally, $\Ptime = \PPAD$ is the strongest collapse we can hope to prove in the conclusion of this statement.  However, intermediate collapses, such as $\CLS = \PPAD$ (as in the previous theorem), are also very much worthy goals. This conjecture is the natural one to make given Theorem~\ref{CLS theorem}, yet it appears to be much more difficult to prove.

We believe that proving the Impossibility Conjecture would be a key step on the way of understanding game dynamics, as well as their connections to Complexity Theory.  Further, we believe that this conjecture will require for its proof a new style of mathematical reasoning, a sophisticated methodology that draws from Complexity Theory, Game Theory, and the topological theory of dynamical systems.  What is required here is a novel mathematical theory in which algorithms can be seen as topological objects.  

Next, we define a framework which, we believe, goes a long way towards capturing this difficulty formally. 

\subsection{An Obstacle to Proving the Impossibility Conjecture}

A black-box proof of the Impossibility Conjecture, if one existed, would proceed as follows. Let $\phi_G$ be a hypothesized locally tractable Nash-convergent dynamic, defined for each game $G$. Suppose that, for any game $G$, we could query $\phi_G$ at some sequence of points in polynomial time --- where the next query can be computed in polynomial time from the answers to the previous ones --- in such a way that, after a polynomial number of queries, we can compute a Nash equilibrium of $G$ based on the answers. If such an algorithm were proved to exist, we would conclude that $\Ptime = \PPAD$, proving the conjecture.

In this section we will argue that no such proof strategy will work, for the following intuitive reason: there are Nash-convergent dynamics that provide \emph{no information whatsoever} about the Nash equilibria within polynomially many black-box queries. %Therefore it is impossible to deduce whether the dynamic must be locally intractable. %Hence proving this conjecture will require assuming some kind of broader structure on the dynamic (as we did in Section [ref]).

\subsubsection*{\bf The Proving Game.}  First recall \cite{chen2006settling} that there is a polynomial $r$ such that the problem of finding a $\frac{1}{r(|G|)}$-approximate Nash equilibrium in a given two-player game $G$ is PPAD-complete.  We paraphrase this as follows:

\begin{lemma} 
If $\PPAD \neq \Ptime$, then for every polynomial-time algorithm there are infinitely many two-player games $G$ for which the algorithm will fail to compute a $\frac{1}{r(|G|)}$-approximate Nash equilibrium.  
\end{lemma}

We next define a game denoted {\bf PG,} played between a polynomially-bounded {\em Black-box prover} Bob and a computationally unbounded {\em Adversary} Alice. To start, Bob chooses a polynomial $p$. Then a two-player game $G$ is chosen, and the following {\em run on game $G$} of {\bf PG} is played: Alice claims to possess a Nash-convergent discrete-time dynamic in $G$, and she challenges Bob to locate a Nash equilibrium of $G$ by querying her at points of the mixed strategy space of $G$. 
%Such a game is not straightforward to define, however, for the familiar reason that it is not clear what it means for Bob to come up with a Nash equilibrium in a game with more than two players, as such equilibria are known to have large complexity \cite{EY}.  

Bob and Alice play $p(|G|)$ rounds, during each of which Bob chooses, after a $p(|G|)$-time computation, a point $x$ in the domain $X$ of $G$ (presumably informed by Alice's answers so far), and Alice responds by revealing the step $\phi(x)$ of the claimed dynamic. If after the $p(|G|)$ rounds Bob {\em cannot} reveal a $\frac{1}{r(|G|)}$-approximate Nash equilibrium of $G$, and Alice {\em can} describe a legitimate Nash-convergent dynamic that is consistent with all of Alice's previous replies to the queries by Bob, then Alice {\em prevails} in the run of {\bf PG} on $G$; otherwise Bob prevails.  

Finally, Bob wins {\bf PG} if he prevails {\em in all runs}, for all games $G$.  Otherwise, Alice wins.  

We believe that {\bf PG} as defined above captures reasonably well the process of proving the Impossibility Conjecture by a black-box proof.
Hence, the following result can be seen as evidence that discovering a black-box proof of the Impossibility Conjecture is problematic and unlikely:

\begin{theorem}
If $\Ptime \neq \PPAD$, {\bf PG} is a forced win by Alice.
\end{theorem}

\paragraph{Remark:} Recall that {\bf PG} is of interest only if $\Ptime \neq \PPAD$, because the Impossibility Conjecture would be false if $\Ptime = \PPAD$ via Type 1 dynamics.
\begin{proof}
We shall describe a winning strategy for Alice. %The idea of Alice's strategy is as follows: For any game $G$, Alice will build a Nash-convergent dynamic $\Phi$ that is a mixture of our Type I dynamic $\phi_N$ (`go directly to a fixed Nash equilibrium')  and the flow of a standard game dynamic (the Brown-von Neumann-Nash dynamic) $\phi_D$ \cite{??}. 
In the run of the {\bf PG} on a game $G$, in the first, interactive stage of the Proving Game, %during which Bob chooses points in $X$, 
Alice responds to each of Bob's queries $x$ with $\phi(x)= x+ \gamma {\beta(x)}$. Here $\gamma$ is an exponentially small quantity in the size of the game, $\gamma = f(|G|)2^{-|G|}$, where $f$ is some polynomial that arises as a function of $r(|G|)$. The vector $\beta$ is the direction of a standard game dynamic called the {\em Brown-von Neumann-Nash dynamic}, defined for each strategy $a$ of player $i$ as follows \cite{sandholm2010population}:
\[
\beta(x)_a = \bigl[u_i(a, x_{-i}) - \bar u_i(x)\bigr]_+
-
x_i^a \sum_{b}
\bigl[u_i(b, x_{-i}) - \bar u_i(x)\bigr]_+
\]

\noindent Here, $u_i$ is the utility function of player $i$ and $\bar u_i$ the average utility weighted over all strategies, while the $[.]_+$ notation takes the max between zero and the argument.  Intuitively, this is a best response dynamic in which each player moves probability mass to the strategies that have above average utility. Importantly, its fixpoints are exactly the Nash equilibria. %Note that the precise choice of the dynamic is not important, and an array of alternatives would also work.

It is important to note that this gives Bob no new information, because Bob could compute $\phi(x)$ in polynomial time without input from Alice. In addition, the fixed points of this dynamic are exactly the Nash equilibria \cite{sandholm2010population}, so Bob cannot compute a fixed point in polynomial time.
%(several other choices of directions would work here in place of $\dot\beta$).   

Now notice that the interaction between Alice and Bob so far can be seen as a candidate deterministic algorithm with the game $G$ as input: the algorithm can output the polynomially many query points put forward by Bob, and check whether one of them is a $\frac{1}{r(|G|)}$-approximate Nash equilibrium.   {\em Our proof of the theorem will focus on the run of the {\bf PG} on the particular game $G$, guaranteed by Lemma 5.1, on which this algorithm fails,} and will establish that Alice will prevail in this game, completing the proof of the theorem.  Among the infinitely many such games, we shall choose one that is large enough so that certain asymptotic inequalities hold. This simplifies the analysis considerably.

Let $Q := \{x_1,x_2,\dots,x_k\}$ to be the set of \emph{query points} by Bob, with $|Q| \leq p(|G|)$. Similarly, let $\phi(Q)$ denote the set of successors of $Q$ under, that is $\phi(q_i)$ for all $q\in Q$. Note that Bob can compute all of $Q\cup \phi(Q)$ in polynomial time. We will use the following two facts about $Q$:

\begin{enumerate}
\item $Q\cup \phi(Q)$ contains no $\frac{1}{r(|G|)}$-approximate Nash equilibrium;

\item There is no pair of points $q,q'$ in $Q \cup \phi(Q)$ such that there is a Nash equilibrium within a distance of $2^{-2|G|}$ from the line between $q$ and $q'$.
%\item $Q$ contains no points $q,q'$ such that there is a Nash equilibrium within a distance of $2^{-2|G|}$ from the line defined by $q$ and $q'$.
\end{enumerate}

The first claim follows directly from the assumption that $\Ptime\neq \PPAD$. The second results from the following lemma:

\begin{lemma}
Bob cannot compute a pair of points $y$ and $z$ in polynomial time where a Nash equilibrium lies within a distance of $2^{-2|G|}$ from the line between $x$ and $y$.
\end{lemma}
\begin{proof}
    Suppose that the line defined by $y$ and $z$ is $y+b\cdot\lambda$, and write the symmetrized equilibrium conditions $(Ax)_i \le x^T A x, i=1,\ldots,n$ for $G$, where $n$ is the number of strategies of $G$, in terms of the variable $\lambda$. The resulting system of quadratic inequalities in $\lambda$ can be rewritten as simple ordering conditions on the sorted roots of the resulting $n$ quadratic polynomials in $\lambda$. If there is a Nash equilibrium within $2^{-2|G|}$ of this line, then these inequalities would have a near-solution, one that violates some of them by at most $2^{-|G|}$ for some $\lambda$, and Bob can determine this near-solution by inspection of the inequalities.  That would give Bob a $\frac{1}{r(|G|)}$-approximate Nash equilibrium of $G$. 
\end{proof}

To define the required dynamic $\Phi$, we start by fixing $\Phi(x) = \phi(x)$ for every $x\in Q$. Next, we enclose each query point $q$ with a ball $B_q$ with center $q$ and radius $\gamma^2$ --- exponentially smaller than the step of $\phi$. Outside these balls, we define $\Phi(x)=x+(T-x)$ as a discrete dynamic of Type 1, which maps each point directly to a fixed Nash equilibrium which we call $T$. 

Now suppose that $x\neq q$ is a point inside $B_q$, and let $x'$ be the point on the boundary of $B_q$ that is closest to $x$. Then we define
\[
\Phi(x) = x + \epsilon(T-x') +(1-\epsilon_i)\gamma\beta(q)
\]
where $\epsilon_i = |x - q|/\gamma^2$ is the proximity to the center of the ball. Notice that, in general, the $B_q$ balls may overlap. In such a case, we will define $\Delta(x)$ by the following procedure. First, we partition $Q$ into subsets $K_i$, initially the singletons. Let $S_i$ be the smallest ball containing all points in $K_i$ --- we merge $K_i$ and and $K_j$ if these balls lie within $2\gamma^2$ of each other. Eventually we arrive at a partition of points where these spheres are disjoint, separated by at least $2\gamma^2$, and the diameter of each sphere is bounded by a polynomial in $\gamma^2$. Then finally we define $\Phi(x)$ as $x + \gamma \beta(x)$ within $S_i$, and in the $\gamma^2$ neighborhood of $S_i$ we interpolate with the type 1 dynamic as before. Because this ball is still small (a polynomial multiple of $\gamma^2$), the asymptotic analysis is the same, so for the sake of simplicity we shall proceed assuming that the balls $B_q$ do not overlap.

To prove that this is a legitimate dynamic, we first note that it is the linear combination of two continuous dynamics, and so it is continuous. The challenge of the proof is showing that all points converge to Nash equilibria --- we will show in fact that all points converge to $T$. By compactness, if any point does not converge to a Nash equilibrium, then there exists some non-Nash point which returns arbitrarily close to itself. We will rule out the possibility of such cyclic behavior.

Suppose for contradiction that there exists a set of points $x_1,\dots,x_k$, which form a $2\gamma^2$-\emph{approximate cycle}, that is $\Phi(x_i)$ lies within $2\gamma^2$ of $x_{i+1}$ for each $i$ (modulo $k$). We will show that the existence of such a cycle must violate one of our claims above about the query set $Q$.

First, observe that every point $x_i$ on this approximate cycle must lie within some ball $B_q$, because otherwise $\Phi(x_i) = T$. We write $q_i$ for the query point which lies within $\gamma^2$ of $x_i$. Secondly, because the radius of each ball is $\gamma^2$, we can assume that each ball contains at most two points on this cycle --- otherwise there is another shorter $\gamma^2$-approximate cycle. Hence this cycle contains at most $2p(|G|)$ points, all of which lie within $\gamma^2$ of some query point. We can express the cycle condition as the following equality
\begin{equation} \label{cycle equation}
    \sum_{i} (\epsilon_i (T-x_i') + (1-\epsilon_i)\beta(q_i)\gamma + \alpha_i)  = 0
\end{equation}
where $\alpha_i \leq \gamma^2$ are small vectors representing the difference between $\Phi(x_i)$ and $x_{i+1}$.

\emph{Claim:} For each $i$, $\epsilon_i \leq \poly(|G|) \gamma$.

Using \eqref{cycle equation}, for any $i$, $\|\epsilon_i(T-x_i')\| \leq \sum_i\|(1-\epsilon_i)\gamma\beta(q_i) + \alpha_i\| \leq \gamma \sum_i \|\beta(q_i)\| + \gamma^2 \leq \gamma \poly(|G|)$. Similarly, because $q_i$ is not approximate Nash, $\|T-x_i'\| \geq |\|T-q_i\| - \gamma^2| \geq \poly(|G|)$. Hence $\epsilon_i \leq \gamma \poly(|G|)$. In other words, every $\epsilon$ must be very small, so each $x_i$ is very close to $q_i$ within the ball. The ball has radius $\gamma^2$, so $\|x_i - q_i\|\leq \gamma^3\poly(|G|)$.

\emph{Claim:} There exists some $i$ and $k$ where $|\beta(x_i)_k |\geq \poly(|G|)$.

Define the \emph{relative deviation} $D_s^1(x,y) = (Ay)_s - x^TAy$ for player 1, with an analogous definition for player 2. A point $z = (x,y)$ is an $r(|G|)^{-1}$-approximate Nash equilibrium if each entry of $D_s(z)$ is bounded by $r(|G|)^{-1}$. Suppose that $D_s(z)$ is polynomially large. If $x_s$ is small (subpolynomial), then $\beta_s(z) = [D_s(z)]_+ - x_s\sum_r [D_r(z)]_+$ is at least polynomially large. On the other hand, if $x_s$ is also at least inverse polynomial, then $x_sD_s$ is also polynomial. Now observe that $\sum_i x_i D_i = 0$, and so there exists an index $r$ where $x_r D_r$ is a negative quantity of polynomial magnitude. Both terms are bounded, so each has at least polynomial magnitude. Further $x_r$ is positive, and so $D_r(z)$ is negative. Hence $|\beta_r(z)| = |[D_r(z)]_+ - x_r\sum_i [D_i(z)]_+| = |-x_r\sum_i [D_i(z)]_+| \geq x_rD_s(z)$  which is at least polynomial.

\emph{Claim:} Some $\epsilon_j \geq \gamma\poly(G)$.

We want to show that some $\epsilon$ is at least polynomially large compared to $\gamma$. We will use the fact that no query point is an approximate Nash equilibrium, and so some $|\beta(x_i)_k|$ is polynomially large. First, observe that the magnitude of each step is a polynomial function of $\gamma$, and so is exponentially small. Further, because $\beta$ is a polynomial function of $x$, the total change in $\beta$ over this polynomial-sized set of points is a polynomial in $\gamma$. In particular, if $|\beta(x_i)_k|$ is polynomially large for any point $x_i$, then $\beta(x_j)_k$ has the same sign for all points in this set.

Using \eqref{cycle equation}, we get that
\begin{align*}
    \gamma|\beta(x_i)_k| &\leq \sum_i \epsilon_i |(T-x_i')_k| + \sum_i\gamma^2 \\ 
    \gamma|\beta(x_i)_k| - p(|G|)\gamma^2 &\leq \sum_i \epsilon_i |(T-x_i')_k| \leq \sum_i \epsilon_i \\ 
\end{align*}
There are polynomially many points $x_i$, hence some $\epsilon_i$ must have magnitude at least $\gamma \poly(|G|)$, completing the claim.

To complete the proof, we will show for a contradiction that the Nash equilibrium $T$ lies approximately (within $2^{-2|G|}$) on the line between $\phi(q_i)$ and $q_{i+1}$, for the index $i$ where $\epsilon_i \geq\gamma \poly(|G|)$.

%Observe that $|T - q_{i+1}|\geq r(|G|)^{-1}$, because every query point is not approximate Nash. 
Observe that $x_{i+1} := \Phi(x_i) = x_i + \epsilon_i(T-x_i') + (1-\epsilon_i)(\phi(q_i) - q_i) = x_i-q_i + \epsilon_i(x_i' - q_i) + \epsilon_i T + (1-\epsilon_i)\phi(q_i)$. Further, recall that $x_{i+1}$ lies within $\epsilon^3\poly(|G|) $ of $q_{i+1}$. Using similar triangles, the distance between this line and $T$ is bounded by the distance between $q_{j+1}$ and the line between $\phi(q_j)$ and $T$, scaled by the ratio of the distances from $\phi(q_i)$ to $q$ and $T$ respectively. 

We know that $z := \epsilon_i T + (1-\epsilon_i)\phi(q_i)$ is a point lying on the line between $\phi(q_i)$ and $T$, $x_{i+1} - z= x_i-q_i + \epsilon_i(x_i' - q_i)$, which is bounded by $2\gamma^3\poly(|G|)$. Further, the distance between $\phi(q_i)$ and $T$ is at most a constant. Finally, the distance between $\phi(q_i)$ and $q_{i+1}$ can be lower-bounded by

\begin{align*}
    \|\phi(q_i) - q_{i+1}\| &\geq |\|\phi(q_i) - x_{i+1}\| - \|x_{i+1}-q_{i+1}\||  \geq |\|\phi(q_i) - x_{i+1}\| - \gamma^3\poly(|G|)| \\
    &\geq |\|\phi(q_i) - x_i +q_i - \epsilon_i(x_i' - q_i) - \epsilon_i T - (1-\epsilon_i)\phi(q_i)\| - \gamma^3\poly(|G|)| \\
    &\geq |\epsilon_i\|\phi(q_i) - T\| - \gamma^3\poly(|G|)| \\
    & \geq \gamma \poly(|G|)
\end{align*}
where we used the fact that $\epsilon_i \approx \gamma \poly(|G|)$.

Finally, we get that the distance between $T$ and the line from $\phi(q_i)$ to $q_{i+1}$ is at most $2\gamma^3\poly(|G|) / (\gamma \poly(|G|)) = \gamma^2\poly(|G|)$. We picked $\gamma $ to be a polynomial function multiplied by $2^{-|G|}$, and we choose the polynomial so that this expression is at most $2^{-2|G|}$. This contradicts our claim about $Q$, completing the proof.

\end{proof}

\section{Discussion}

All known locally tractable dynamics are not known to be Nash-convergent; all known Nash-convergent dynamics, such as those we discussed here, are not locally tractable and are therefore unsuitable as models of player behavior. 
%We showed that Nash-convergent dynamics exist in all games,  implying that degeneracy is necessary for the Impossibility Theorem \cite{milionis_impossibility_2023} to hold. However, all such known dynamics are intractable to compute locally, and therefore unsuitable as models of player behavior.
The overarching remaining open question is:

\begin{conjecture}
(The Impossibility Conjecture) If for all games there is a locally tractable Nash-convergent dynamic, then $\PPAD = \Ptime$.  
\end{conjecture}
One of the main results of this paper is the proof that a uniformly Nash-convergent dynamic cannot exist if it also possesses a locally tractable Lyapunov function, conditional on the weaker assumption $\PPAD \neq \CLS$.  

We believe that the Impossibility Conjecture is very plausible, chiefly because of the decades of the pursuit by a large community of economists, mathematicians, and computer scientists of an efficient and Nash converging dynamic in games. Proving it is an important research goal to pursue --- however, the Proving Game casts doubts on the likelihood of a proof appearing soon.  

Taken together, the preceding discussion summarizes fresh evidence that Nash learning is not an easy matter, even in the absence of degeneracy.  In view of this, we believe that the Nash equilibrium becomes even less attractive than we knew before, and alternative  solution concepts in Game Theory should be considered, presumably ones that are compatible with and informed by computation. 
The solid motivation and rich theory of dynamical systems as models of the behavior of rational agents playing a game point to such an approach, first anticipated in \cite{papadimitriou_game_2019}: Adopt the replicator dynamic \cite{sandholm2010population} --- gradient descent on the players' utilities, or equivalently the small step limit of the multiplicative weights update algorithm \cite{arora_multiplicative_2012} --- as the natural model of behavior of the players, and define the meaning of the game as the limit of this behavior: 

\medskip
\noindent{\em
The meaning of a game $G$ is a function $\mu_G$ that maps any prior distribution over the mixed strategy profiles of $G$ to the limit distribution over the attractors of the replicator dynamic in $G$.
}
\medskip 

In other words, the meaning of a game is the way in which repeated play of the game transforms the correlated behavior of the players.  This point of view motivates the following two important, and challenging, mathematical problems, which we believe are central to the development of a computational theory of learning dynamics:

\begin{conjecture}
There is a polynomial-time algorithm which, given any game $G$, outputs a description of the attractors of the replicator dynamics in $G$.  

\end{conjecture}

\begin{conjecture}
There is a polynomial-time algorithm which, given a game $G$ and a mixed strategy $x$, computes the limit attractor of the replicator dynamic starting at $x$. 

\end{conjecture}

Resolving these two important problems would provide a comprehensive, principled, and computationally accessible theory of the meaning of the game, radically reframing the bedrock of Game Theory in dynamical terms.  See \cite{biggar2026sink,biggar2026computing,biggar_attractor_2024,biggar_replicator_2023} for recent progress towards proving the first conjecture, and \cite{hakim2024swim} for the algorithmic solution of an approximate version of the second.

%\paragraph{Acknowledgment:} We thank Ruta Mehta and Mihalis Yannakakis for insightful discussions.  This research was supported by the National Science Foundation.

\bibliographystyle{alpha}
\bibliography{references}

\end{document}